\documentclass[noinfoline, preprint, 12]{imsart}

\usepackage{endnotes}
\usepackage{pgfplots}
\usepackage{multirow}

\usepackage{amsfonts}
\usepackage{amsthm}
\usepackage[cmex10]{amsmath}
\usepackage{amssymb}

\let\footnote=\endnote

\newtheorem{theorem}{Theorem}
\newtheorem{definition}{Definition}
\newtheorem{proposition}{Proposition}

\usepackage{verbatim}
\usepackage{bbm}

\begin{document}

\begin{frontmatter}

\title{A Customer Choice Model with HALO Effect}
\runtitle{Halo Effect Model}

\author{Reza Yousefi Maragheh, Alexandra Chronopoulou}

\address{Department of Industrial \& Enterprise Systems Engineering\\ University of Illinois at Urbana-Champaign, Urbana, IL 61801\\
\texttt{ysfmrgh2@illinois.edu, achronop@illinois.edu}}
\author{James Mario Davis}
\address{Uber, Marketplace Reseach Group, San Fransisco, CA\\
\texttt{jamesmariodavis@gmail.com}}


\begin{abstract}
In this paper, we propose an extension to the multinomial logit (MNL) model, the Halo MNL, that takes into account the interaction effects among  products in an assortment. In particular, this model incorporates pairwise interactions of items in an effort to describe positive/negative effects among products that are present/absent in the assortment. Furthermore, we are interested in establishing sufficient conditions for identifiability, in order to build robust estimation methods. Under strict identifiability conditions, we use maximum likelihood to estimate the model parameters for which we derive closed formulas. We also perform simulation experiments, in order to numerically evaluate our method,  study the accuracy of the estimators and  compare it with the MNL. Last, we fit our model in the Hotel Chain dataset in \cite{Bodea2009Choice-BasedChain}, and we compare it with  MNL in terms of efficiency, accuracy and robustness. We conclude that for rich enough datasets the model that includes interaction effects  performs better in terms of how well it fits the data.
\end{abstract}


\end{frontmatter}

\section{Introduction}
Many industries, including online retailers, travel and transportation companies, are interested in selecting the optimal subset of \textit{products} or \textit{items}, to be presented to the customer, with the goal to maximize the acquired revenue. In order to solve this \textit{optimal assortment} problem, it is very  important to properly model the \textit{demand} of each product or in other words the customer's behavior in selecting it (customer's preference). Under different formulations, finding the optimal assortment boils down to a non trival optimization problem that has been addressed in the literature under different formulations, for example in \cite{Talluri2004RevenueBehavior, topal1, topal2}.

The multinomial logit model (MNL) is one of the most popular customer choice models in the literature, according to which a  weight is assigned to every item denoting the customer's preference relative to the other items, see for example in \cite{Train2009DiscreteSimulation}. However, in practice, not all products are available to the customer at all times, either due to stock-outs,  company's marketing decisions or  due to feasibility, when for example there are space restrictions.  As a result, the customer  is only able to select from an often limited subset of products.  It is thus natural to regard customer choices as a \textit{function} of the available items at the time of purchase. This implies that a customer's preferences are governed by an MNL \textit{conditional} on the  products presented to her, also called the \textit{offer} or \textit{availability set}.

One of the main properties of the MNL is the Independence of Irrelevant Alternatives (IIA), that is the assumption that a customer's choice between two alternative products is \textit{unaffected} by the presence of other products. However, it has been documented, \cite{Berlyne1960ConflictCuriosity, Kahn1991ModelingAssortments, Kahn1995ConsumerReview, Kalyanam2007DeconstructingContribution, Borle2005TheRetention}, that this is not the case in practice. On one hand, it has been observed that the effect of an item on the others can be \textit{positive}, since every additional item may add to the variety of the assortment making the customer more confident to buy, \cite{Kreps1979AFlexibility, Hoch1999TheAssortment, Helson1964Adaptation-levelTheory, McAlister1982VarietyReview}. On the other hand, the effect of an item may also be \textit{negative}, since it is possible to contribute to creating clutter making the customer's decision process harder, \cite{Boatwright2001ReducingApproach, Borle2005TheRetention, Iyengar2000WhenThing}.

In this paper, we propose an extension to the MNL that captures pairwise effects  of an item's availability on customer preferences. Specifically, we introduce two-way interaction terms in the linear MNL model, which implies   that we  no longer assume IIA shifts in customer preferences. From a statistical point of view, this additional  flexibility  increases the number of model parameters that need to be estimated, and affects the model's identifiability, that is the ability to properly estimate the  parameters with the available data.  The identifiability question does not directly follow from established results in the  literature about Generalized Linear Models, due to the unique features of  this framework.

Therefore, our  goals in this paper are the following: first, we want to develop an efficient estimation procedure for  the parameters of the proposed model, second  obtain sufficient identifiability conditions that guarantee the robustness of our approach, and third fit the model in a real dataset.

\section{Related Literature}
Estimation questions for customer choice models with different levels and types of product availability have already been addressed in the literature in different contexts. In this Section, we plan to review the related literature with particular focus on papers that propose estimation methods for demand and customer choice models when the product availability is not constant.

One of the  first papers that accounts for changes in  product availability is the work in \cite{Anupindi1998EstimationProducts}. In particular, the authors consider the case where the customer's first choice is missing  leading to product substitution, and propose an Expectation Maximization (EM) algorithm to estimate the substitution probabilities. In a much more general setting, the authors in  \cite{Talluri2004RevenueBehavior}  propose an EM algorithm for estimating the parameters of a  general customer choice model with varying offer sets. Under an Multinomial Logit  (MNL) customer choice model, in \cite{Vulcano2012EstimatingData}  they  propose an EM algorithm for estimating the lost demand as well as the  model parameters, while in \cite{Lee2016EstimatingModel}  the authors extend this work  to a nested-MNL choice model. 

The EM algorithm is a powerful technique that is adopted in this framework in order to deal with missing information.  In particular, when an item is not present in the offer set, it is not possible to be selected by the customer, and as a consequence it is not feasible to obtain an observation at that point. However, the main difference with more traditional missing data problems is the fact that the data is not missing at random, since their presence/absebce from the assortment is embedded in its selection. A naive technique, when there are missing observations, is to ignore this fact and work with  aggregate data. However,   using simulation experiments, in \cite{Bruno2008ResearchAvailability} the authors concluded that treating the data set as complete can lead to incorrect estimates for the customer choices. 

Another way to deal with missing data is to impute the missing information using what is available. In a slightly different context and  for a demand model that accounts for stock-outs, in \cite{Musalem2010StructuralOut-of-Stocks} they   propose a simulation method  to approximate the set of products that are available to the customer (on the shelves) and impute the missing information  using aggregate sales data . An additional feature of their method is that it updates periodically, which gives the scheme an adaptive flavor. In similar spirit of periodic updates, in \cite{Conlon2013DemandAvailability}, they  develop an  EM algorithm to capture the changes in the offer sets between two consecutive updates. However, the proposed approach suffers from the exponential number of parameters involved in the steps of the algorithm which limits its applicability.  In a dynamic framework that also accounts for the time that a stock-out happens, the authors in \cite{Jain2015DemandNeed} investigate the information gained from knowing the time that an item is missing and its effect on product sales.

Apart from the statistical challenges that are caused due to the variation of the offer sets,  missing products have significant  implications on the customer's choices and behavior when the desirable item is not present. From the seller's point of view,  the customer will ideally substitute this item with another in the proposed assortment. However, she has always the option to leave the ``store'', which is the worst case scenario. In this direction, in \cite{Kalyanam2007DeconstructingContribution} they  investigate the effects of unavailability of items on so-called category contribution, which is the impact of the presence of an individual item on the sales of the whole category. Based on simulations, they conclude that the items with the most dominant category impact are not necessarily the top sellers.

In this paper, we propose an extension to the MNL model that takes into account the interaction between the available and missing product in an assortment, the Halo MNL. For this new model, we are interested in establishing conditions for identifiability as well as suitable conditions for partial identifiability, in order to build robust estimation methods. Under strict identifiability conditions, we use maximum likelihood to estimate the model parameters for which we derive closed formulas. 

We also present both simulation experiments for the performance of our model, as well as a real data example. In particular, we simulate data from a Mixed MNL model, and we fit both the MNL and our model. Although the number of parameters in the Halo-MNL is significantly larger, we find that it has a better fit to the data in terms of AIC/BIC scores, which are metrics that penalize a model for the higher number of parameters. Furthermore, we apply both the Halo-MNL and MNL models to a Hotel dataset,  \cite{Bodea2009Choice-BasedChain}. In this real data example,  not only do we observe a better fit of the Halo-MNL to the data, but we also derive conclusions regarding the effects of missing items on the remaining ones in the assortment.

In the sequel, the structure of our paper is as follows: In Section 3, we describe our proposed model. In Section 4, we discuss the maximum likelihood approach for our model, as well as sufficient identifiability and partial identifiability conditions. In Section 5, we present our numerical results both from simulation experiments and from the fitting of our model to a Hotel dataset. We discuss our main conclusions and future work in Section 6.

\section{Model Description}

Consider a collection of $\mathcal{N} = \{1, \ldots, N\}$ products to be sold over $M$ periods and denote  by $\mathcal{S}_{m}\subset \mathcal{N}$ the set of items offered to the customer in a given period $m\in M$. We assume that, due to limited availability for example, not all products are available  on a given period $m\in M$, which implies that for different periods the assortment sets $\mathcal{S}_{m}$ may be different. Moreover, in a given period $m\in M$, the customer may choose  one item from the assortment or nothing, the so-called no-purchase option.  The offer set $\mathcal{S}_{m}$ is assumed to be known to us in advance, while our observations consist of the customer choice, i.e. purchase, in period  $m$, denoted by  $z_{i}^{(m)}$, $i \in \mathcal{N}$. If $i\notin \mathcal{S}_m$ then trivially  $z_{i}^{(m)}=0$. The aggregate number of items of type $i$ bought in all periods is simply denoted by $z_{i} = \sum_{m\in M} z_{i}^{(m)}$.

In the  literature,  the most popular model to describe the customer choices is the multinomial logit model (MNL). According to this model, each customer has a preference for each item denoted by $v_i$, $i=1,\ldots, N$, with $v_0=1$ being the normalized no-purchase preference.  Therefore, the probability that a customer selects product $j\in \mathcal{S}_m$, when $\mathcal{S}_m$ is offered,  is defined as 
\begin{equation}\label{eq:classicMNL}
P_j(\mathcal{S}_m, v) = \frac{e^{v_j}}{\sum_{i\in \mathcal{S}_m} e^{v_i} +1},\quad \text{with   }\quad
P_j(\mathcal{S}_m, v) = 0, \text{ when } j \notin \mathcal{S}_m
\end{equation}
where $v$ is the vector of customer preferences for the $N$ items.

In this paper,  we propose a model that captures the ``halo'' effect, which is defined as the effect  of one product being present or absent in the offer set on the customer's preference/action. Therefore,  we decompose  the customer's preference for item $j$ in two components: \textit{(i)} the baseline preference (i.e. the preference when all items are offered) for item $j$, denoted by $\mu_j$, and \textit{(ii)} the effect of the absence of item $i$ on the preference for item $j$, denoted by  $\alpha_{ij}, \,\forall i,j =1, \ldots, N$, with $\alpha_{ii}=0$. Specifically, we model the term $v_j$ in \eqref{eq:classicMNL} as
\begin{equation}\label{eq:MNL_halo}
v_{j}=x_{j}^{(m)}\mu_{j}+\sum_{i=1}^{N}x_{j}^{(m)}(1-x_{i}^{(m)} )a_{ij},
\end{equation}
where  $x^{(m)}=(x_{1}^{(m)},x_{2}^{(m)},...,x_{n}^{(m)})$ is the availability vector during period $m$  with $x_{i}^{(m)}= 1$, if $i\in \mathcal{S}_m$ and 0, if $i\in\mathcal{S}_{m}$.  Therefore,  the probability of purchasing  item $j$ in period $m$ becomes:
\begin{align*}
P_{j}^{(m)} &=\frac{\exp\left(x_{j}^{(m)}\mu _{j}+\sum_{i=1}^{N}x_{j}^{(m)}\left(1-x_{i}^{(m)}\right)\alpha_{ij}\right)}{1+\sum_{j=1}^{N}\exp{\left(x_{j}^{(m)}\mu _{j}+\sum_{i=1}^{N}x_{j}^{(m)}\left(1-x_{i}^{(m)}\right)\alpha_{ij}\right)}},\\
 P_{0}^{(m)} &=\biggl[1+\sum_{j=1}^{N}\exp \Bigl(x_{j}^{(m)}\mu_{j}+\sum_{i=1}^{N}x_{j}^{(m)}\Bigl(1-{x}_{i}^{(m)}\Bigr)  \alpha_{ij}\Bigr)\biggr]^{-1},
\end{align*}
which can also be written as 
\begin{equation}\label{eq:probabilities}
P_{j}^{(m)}=\frac{\exp\left(\mu _{j}+\sum_{i\notin {S}_{m}}\alpha_{ij}\right)}{1+\sum_{j\in S_{m}}\exp\left(\mu _{j}+\sum_{i\notin S_{m}}\alpha_{ij}\right)},\quad P_{0}^{(m)}=\biggl[1+\sum_{j\in S_{m}}\exp\Bigl(\mu _{j}+\sum_{i\notin S_{m}}\alpha_{ij}\Bigr) \biggr]^{-1}
\end{equation}
with the preference of the no-purchase item be normalized to 1.

\section{Parameter Estimation}

The goal in this section is to describe an estimation procedure and discuss identifiability conditions for the model parameters $\mu_j$, $\alpha_{ij}$, for $i,j=1,\ldots, n$. First, let us  define the matrix containing all parameters to estimate, $\Theta$, as follows:
\begin{equation}
\Theta:= 
\begin{bmatrix}
\mu_1 & \mu_2 & \ldots & \mu_{N-1} & \mu_N \\
0 & \alpha_{12} & \ldots & \alpha_{1,N-1} & \alpha_{1N}\\
\alpha_{21} & 0 & \ddots & \dots& \vdots\\
\alpha_{31} & \alpha_{32} & \ddots & \ddots& \vdots\\
\vdots & \vdots & \ddots & 0 &\alpha_{N-1,N}\\
\alpha_{N1} & \alpha_{N2} & \ldots & \alpha_{N,N-1} & 0
\end{bmatrix} = 
\begin{bmatrix}
\mathbf{\mu}\\
\mathbb{A}
\end{bmatrix},
\end{equation}
where $\mathbf{\mu}$ is the vector containing  $\mu_j$, $j=1,\ldots, N$  and $\mathbb{A}=[\alpha_{ij}]$, with $i,j=1,\ldots,N$.

\subsection{Likelihood Function}

Since the underlying model is an MNL-type model, the likelihood function is as follows:
\begin{equation}\label{eq:likelihood}
L(\mu, \mathbb{A}) = \prod_{m=1}^{M}   \left(P_{0}^{(m)}\right)^{ z_{0}^{(m)}}  \left(  \prod_{j\in \mathcal{S}_m} \left(P_{j}^{(m)}\right)^{ z_{j}^{(m)}}  \right).
\end{equation}
If we take the logarithm and plug-in the expression, \eqref{eq:probabilities}, for the probabilities $P_j^{(m)}$ we obtain the following log-likelihood:
\begin{align*}
\ell(\mu, \mathbb{A}) &:= \log L(\mu, \mathbb{A}) = \sum_{m=1}^{M}\left( z_{0}^{(m)} \;\log P_{0}^{(m)}  + \sum_{j\in S_m} z_{j}^{(m)} \log P_{j}^{(m)} \right)\\
&=  \sum_{m=1}^{M}\left\{-\kappa^{(m)}  \;\log\biggl(1+  \sum_{k\in S_{m}}\exp\Bigl( \mu _{k}+\sum_{i\notin S_{m}}a_{ik}\Bigr) \biggr)  +  \sum_{j\in S_m} z_{j}^{(m)}  \Bigl(\mu _{j}+\sum_{i\notin {S}_{m}}a_{ij}\Bigr ) \right\},
\end{align*}
where we set $\kappa^{(m)}:=\sum_{j\in S_m \bigcup \{0\}}z_{j}^{(m)}$.
Taking the derivatives with respect to the parameters, 
\[ \frac{\partial \ell (\mu, \mathbb{A}) }{\partial \mu_p} = 0,\;\;\; \frac{\partial \ell (\mu, \mathbb{A}) }{\partial \alpha_{qp}} = 0, \;\; \forall q, p = 1,\ldots, N,\; \;q\neq p,
\]
we obtain the following system of equations ($\forall  q, p = 1,\ldots, N,\; \;q\neq p$):
\begin{equation}\label{eq:derivativeS}
\sum_{m=1}^{M} \left\{z_p^{(m)} \biggl( 1+ \sum_{k\in \mathcal{S}_m} \exp\Bigl( \mu_k + \sum_{i\notin \mathcal{S}_m} \alpha_{ik}\Bigr) \biggr) - \kappa^{(m)} \exp \Bigr ( \mu_p + \sum_{i\notin \mathcal{S}_m} \alpha_{ip}\Bigr) \right\} = 0
\end{equation}
As one would expect,   \eqref{eq:derivativeS} heavily depends on the structure of the offer sets $\mathcal{S}_{m}$, $m\in M$. In order to better understand the dependence of the likelihood on the offer sets, we define an intermediate matrix $\mathbb{Q}$ that specifies the relationship between the parameters and the offer sets. Specifically, the elements of the $\mathbb{Q}$-matrix indicate the presence or absence of an item in the offered set in period $m$, with $q_{mj} = 1$, if $j\in \mathcal{S}_{m}$ and 0 if $j \notin  \mathcal{S}_{m}$:
\begin{equation}
\mathbb{Q}:= \begin{bmatrix}
q_{11} & q_{12} & \ldots & q_{1N}\\
q_{21} & q_{22} & \ldots & q_{2N}\\
\vdots & \vdots & \ddots & \vdots\\
q_{M1} & q_{M2} & \ldots & q_{MN}
\end{bmatrix}.
\end{equation}
Therefore, the derivative of the log-likelihood can be written as: 
\begin{align}\label{eq:derivative}
\ell^*(\mathbf{\mu},\mathbb{A}) &= \sum_{m=1}^{M} \biggl\{ z_p^{(m)} \biggl( 1+ \sum_{k=1}^{N}  \exp\Bigl( \mu_k \;q_{mk} + \sum_{i=1}^{N} \alpha_{ik}\;(1-q_{mi})\Bigr) \biggr) -  \kappa^{(m)} \exp \Bigr ( \mu_p q_{mp} + \sum_{i=1}^{N}  \alpha_{ip} \;(1-q_{mi})\Bigr) \biggr\}\notag \\
& = 0.
\end{align}

\subsection{Identifiability Conditions}

Due to the nature of the data and the structure of the $\mathbb{Q}$-matrix, it is not always possible to uniquely identify all the model parameters. The most natural example is when, for example, item $i$ is always in the offer set. Then, we are not able to observe its effect when missing on the other items, which is quantified via parameter $\alpha_{ij}, j=1,\ldots, i-1, i+1, \ldots, N$ and $i$ fixed . This, of course, does not mean that there is not an effect, but our data do not allow us to estimate it. 
Mathematically, identifiability is defined as follows:
 \begin{definition}
A set of parameters $\Theta$ is \textit{identifiable} if the following holds:
 \begin{equation}\label{eq:ident}
 P_{j}^{(m)} = P\left(Z = z^{(m)} | \mathbb{Q}, \Theta\right) = P_{j}^{(m)} = P\left(Z = z^{(m)} | \mathbb{Q}, \bar{\Theta}\right) \Leftrightarrow \Theta = \bar{\Theta},
 \end{equation}
 where $z^{(m)}$ are the data in a  period $m$, $\Theta$ is the parameter matrix and $\mathbb{Q}$ is the matrix that indicated the presence/absence of an item in the offer set.
 \end{definition}
The goal in this paper is not to attack this question in its greater generality. However, we want to provide \textit{sufficient identifiability conditions} that will guarantee the uniqueness of the solution of \eqref{eq:derivative}. Therefore, consider the two following cases:

\begin{itemize}
\item[\textbf{(C1)}] Assume that the $\mathbb{Q}$-matrix takes the following form after row-swapping:
\[\mathbb{Q} = \begin{bmatrix}
&\mathbbm{1}\\
&\mathbbm{1}\\
& \mathbbm{1} - \mathbb{I}\\
& \mathbbm{1} - \mathbb{I}\\
& \mathbb{Q}'
\end{bmatrix}.
\]
where $\mathbb{Q}'$ has a similar structure as $\mathbb{Q}$.

\item[\textbf{(C2)}] Suppose $\mathbb{Q}$ has the following structure after row-swapping:
\[\mathbb{Q} = \begin{bmatrix}
&\mathbbm{1}\\
&\mathbbm{1}\\
& \mathbb{U}\\
& \mathbb{U}\\
& \mathbb{Q}'
\end{bmatrix},
\]
where $\mathbb{U}$ is an upper-triangular matrix of ones and  $\mathbb{Q}'$ has a similar structure as $\mathbb{Q}$.
\end{itemize}

\noindent \textit{Remarks:} 
Condition \textbf{(C1) } is slightly stronger than necessary, but we need to guarantee that each offer set appears at least twice so that we have enough information to identify the model parameters.
 Condition \textbf{(C1)} is interpreted as having periods where all items are offered and periods where only one item is missing. In addition, every item must be missing in at least one period.
  Condition \textbf{(C2)} is interpreted as having periods where all items are offered and then periods where sequentially items are not in the presented offer sets. For example, that would mean that after period $m_1$, item $i$ is out-of stock, so it will not be offered in the subsequent period either.

\begin{theorem}
Under the Halo MNL model \eqref{eq:classicMNL} and \eqref{eq:MNL_halo}, if condition \textbf{(C1)} is satisfied, then $\Theta$ is identifiable.
\end{theorem}

\begin{proof}
Given the structure of the $\mathbb{Q}$-matrix, let us rewrite the derivative of the log-likelihood, by splitting the sum in periods  $m=1,\ldots, M_1$ where everything is offered (first part of the $\mathbb{Q}$-matrix) and periods $m=M_1+1,\ldots, M$ 
 where one item is missing (second part of the $\mathbb{Q}$-matrix):
\begin{align*}
\ell^*(\mathbf{\mu},\mathbb{A}) 
& = \sum_{m=1}^{M_1} \biggl\{ z_p^{(m)} \Bigl( 1+ \sum_{k=1}^{N}  e^ {\mu_k}\Bigr) -  \kappa^{(m)} e^{\mu_p}  \biggr\}  +\; \sum_{m=M_1+1}^{M} \biggl\{ z_p^{(m)} \biggl( 1+ \sum_{k=1}^{N}  \exp\Bigl( \mu_k \;q_{mk} + \sum_{i=1, i\neq k}^{N} \alpha_{ik}\Bigr) \biggr) \\
&\quad -  \kappa^{(m)} \exp \Bigr ( \mu_p\; q_{mp} + \sum_{i=1, i\neq p}^{N}  \alpha_{ip} \Bigr) \biggr\}.
\end{align*}
Identifiability, in this case, is guaranteed by solving this system of equations uniquely. Straightforward calculations lead to the following estimators:
\begin{align}\label{eq:mle}
\hat{\mu} &= \ln \left(\dfrac{\sum_{m=1}^{M_1} z_p^{(m)}}{\sum_{m=1}^{M_1} z_0^{(m)}} \right),  \forall p=1,...,N\\
\hat{a_{ip}} &= \ln \left(\dfrac{\sum_{m\in S_i} z_p^{(m)}}{\sum_{m\in S_i} z_0^{(m)}} \right)- \hat{\mu_p}, \forall i,p=1,...,N 
\end{align}

\end{proof}

Now, consider the case that we restrict the parameter space such that $\alpha_{ij}=0$ when $i<j$. This means that certain products only appear in the complete offer sets and then they are out of stock, so they never appear in an offer set. This creates a specific structure of the $\mathbb{Q}$ matrix which allows us to obtain the following  sufficient conditions for identifiability:

\begin{theorem}\label{thm:iden-tri}
Under the Halo MNL model \eqref{eq:classicMNL} and \eqref{eq:MNL_halo},  when the matrix $\mathbb{A}$ is upper triangular and the $\mathbb{Q}$ matrix satisfies Condition \textbf{(C2)}, then the model is identifiable.
\end{theorem}

\begin{proof}

Following the same process as in Theorem 1, we have that 
\begin{align}
\hat{\mu} &= \ln \left(\dfrac{\sum_{m=1}^{M_1} z_p^{(m)}}{\sum_{m=1}^{M_1} z_0^{(m)}} \right), \forall p=1,...,N \\
\hat{a_{ip}} &= \ln \left(\dfrac{\sum_{m\in S_i} z_p^{(m)}}{\sum_{m\in S_i} z_0^{(m)}} \right) - \ln \left(\dfrac{\sum_{m\in S_{i-1}} z_p^{(m)}}{\sum_{m\in S_{i-1}} z_0^{(m)}} \right), \forall i,p=1,...,N, i<p
\end{align}

\end{proof}

A more general criterion for parameter identifiability has been introduced in  the literature, that of \textit{local identifiability}.  According to it, in a local region of the parameter space, there is a unique $\theta_0$ that fits some specified body of data. More formally,
\begin{definition}
A set of parameters $\Theta$ is \textit{locally identifiable}, if there exists a neighborhood $\tilde{\Theta} \subset \Theta$ such that \eqref{eq:ident} is satisfied.
\end{definition}

\begin{proposition}
Under the halo MNL model \eqref{eq:classicMNL} and \eqref{eq:MNL_halo},  when the  $\mathbb{Q}$ matrix satisfies Condition \textbf{(C2)}, then the model is locally identifiable.
\end{proposition}

\begin{proof}
The steps for the proof here are similar to these of Theorem \ref{thm:iden-tri}.

\end{proof}

In practice, this implies that by using a subset of the data that satisfies the desirable structure of the $\mathbb{Q}$ matrix, we are able to uniquely estimate the model parameters in that region of the parameter space.

\section{Numerical Examples \& Application to Real Data}

The goal of this section is to evaluate the performance of the proposed model. First, we present a simulated example to study the convergence of the error estimates. Second, we fit both  the MNL and Halo MNL models   in a hotel chain dataset, \cite{Bodea2009Choice-BasedChain}.

\subsection{Simulation Experiments}

In this section we provide a simulation experiment to study the performance of our parameter estimates. Thus, we simulate the model with parameters $\mathbb{A}$ and $\mu$, as defined in the Appendix. We consider two cases: in the first one, we assume that the $\mathbb{Q}$-matrix satisfies condition \textbf{(C1)}. Therefore, we consider $T$ periods of different sizes $T=100, 200, 400, 800, 1600, 3200$, where for each period we generated 20 replicates letting the total sample size be 10,000. 

Since the generated data satisfy \textbf{(C1)}, we use the closed formulas provided in \eqref{eq:mle} to compute the maximum likelihood estimators. To quantify the performance of the estimators numerically, we compute the absolute relative error, defined as $\epsilon =  |\hat{\theta} - \theta|/\theta$, where here $\theta$  is notation for a generic parameter and $\hat{\theta}$ is the corresponding  estimator.

For illustration purposes, we present the convergence of the absolute error for parameters $\mu_3$ and $\alpha_{43}$ in Figure \ref{fig:error}.

\begin{center}
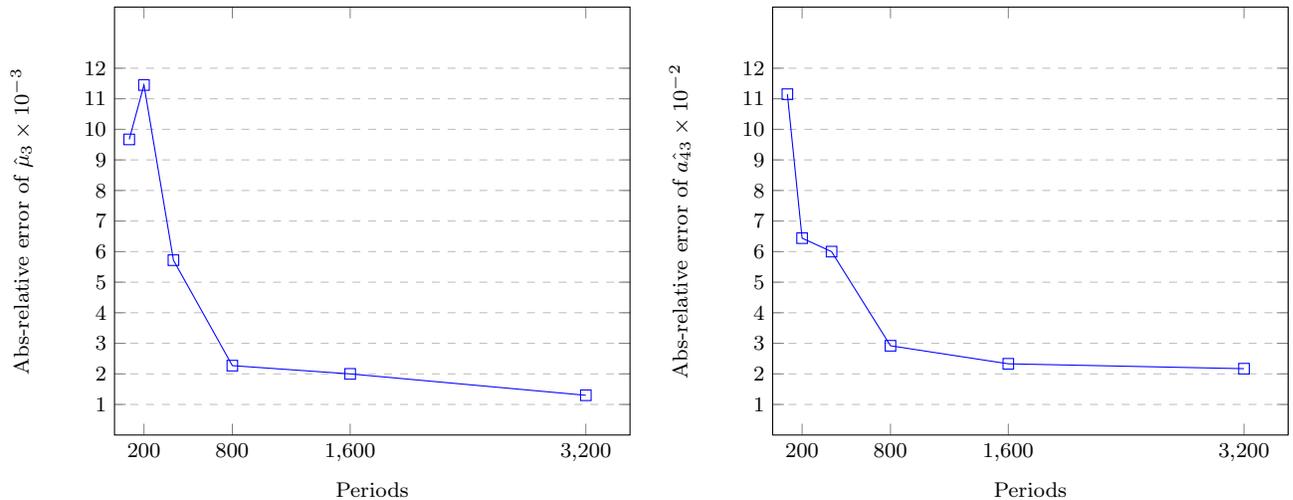
\begin{figure}[h!]\label{fig:error}
\begin{tabular}{cc}
\begin{tikzpicture}
\begin{axis}[
    xlabel={Periods},
    ylabel={Abs-relative error of  ${\hat{\mu }_{3}}\times {{10}^{-3}}$},
    xmin=0, xmax=3500,
    ymin=0, ymax=14,
    xtick={200,800,1600,3200},
    ytick={1,2,3,4,5,6,7,8,9,10,11,12},
    ymajorgrids=true,
    grid style=dashed,
]
 
\addplot[color=blue, mark=square,] coordinates {
    (100,9.67)(200,11.45)(400,5.72)(800,2.27)(1600,2.00)(3200,1.30)};
\end{axis}
\end{tikzpicture}
&

\begin{tikzpicture}
\begin{axis}[
    xlabel={Periods},
    ylabel={Abs-relative error of  $\hat{{a}_{43}}\times{{10}^{-2}}$},
    xmin=0, xmax=3500,
    ymin=0, ymax=14,
    xtick={200,800,1600,3200},
    ytick={1,2,3,4,5,6,7,8,9,10,11,12},
    ymajorgrids=true,
    grid style=dashed,
]
 
\addplot[color=blue, mark=square,] coordinates {
    (100,11.15)(200,6.44)(400,6.00)(800, 2.92)(1600,2.33)(3200,2.17)};
 
\end{axis}
\end{tikzpicture}

\end{tabular}
\caption{Average of Absolute Relative Error as a function of Periods  for parameters $\mu_3$ and $\alpha_{43}$.}
\end{figure}
\end{center}

We are also interested in comparing the performance of Halo-MNL relative to MNL in a simulated environment. In order to make the comparison as fair as possible, we did not simulate data from one or the other model, but we chose to work with the Mixed MNL (MMNL). The MMNL model  is a generalized version of MNL that can be thought as a generator of both models into consideration.

The MMNL model, \cite{McFaddenTrain}, is defined as a MNL model with random coefficients drawn from  a cumulative distribution function defined by $G$ as follows:
\[P_{S}(i|x, \theta) = \int \left( e^{x_i v}/ \sum_{j\in S} e^{x_j v} \right) \cdot G(dv; \theta),\]
where $S$ is the offer set, $x_i$ is as defined before (presence or absence of an item from the offer set), $v$ a vector of random parameters and $\theta$ a vector of parameters characterizing the mixing distribution $G$, denoting the heterogeneity in a population of MNL customers. 
 In our simulation, we assume that there are three types of MNL customers, with fractions  0.2, 0.5, 0.3 and we assume that they arrive at a certain Poisson rate. 

 Since the Halo MNL model includes a second layer of parameters, it is expected that it will have a better performance that the MNL, by capturing the interacting effects of the products, leading to a higher likelihood function score that the MNL. On the other hand, one loses several degrees of freedom, when estimating a larger number of parameters, which eventually reduces the model efficiency due to overfitting.

Therefore, as a metric to compare both models, we use the AIC (Akaike Information Criterion) and BIC (Bayesian Information Criterion) defined as follows:
\[
\text{AIC}  = -2 \;\log \hat{L} + 2\;d,\qquad
\text{BIC}  = -2 \;\log \hat{L} + d\; log n,
\]
where $\hat{L}$ is the maximum of the likelihood function for the model under consideration, $d$ is the number of  parameters to be estimated and $n$ is the sample size. Both criteria are popular metrics for model selection that penalize overfitting (i.e. introducing more parameters than necessary) compared to just the likelihood score. The BIC score, in particular, is  stricter than the AIC, since the penalty term $(d\; \log n)$  depends on the size of the available data. 

In order to make the results easier to interpret, it is  useful to compute the difference of AIC (or BIC) scores for the two models into consideration, e.g. $\Delta \text{AIC} = \text{AIC}_{\text{Model 1}} - \text{AIC}_{\text{Model 2}}$ (or $\Delta \text{BIC}$ respectively). As a rule of thumb, the higher the value of  $\Delta \text{AIC}$  or $\Delta \text{BIC}$, the stronger evidence we have for Model 2. In Table \ref{table:AIC-BIC}, we summarize the results for  $\Delta \text{AIC} =  \text{AIC}_{\text{MNL}} - \text{AIC}_{\text{Halo MNL}} $  and  $\Delta\text{BIC}$, respectively. We observe that in terms of the AIC score, that penalizes only for the number of parameters, the Halo MNL model has in all cases a better fit to the simulated data compared to the MNL. In terms of the BIC score that penalizes for the number of parameters relative to the sample size, for smaller sample sizes the MNL has a better fit to the data, which is not the case when the larger data set. This is expected since in order to fit a model with $n^2$ parameters instead of $n$, a larger dataset is needed.
%
%
%
\begin{table}[h!]
\centering
\caption{\textit{Data Generated by MMN}L:
$\Delta$AIC = AIC$_{\text{MNL}}$-AIC$_{\text{Halo MNL}}$ \& $\Delta$BIC  = BIC$_{\text{MNL}}$-BIC$_{\text{Halo MNL}}$}
\begin{tabular}{|c|cc|cc|cc|cc|cc|}
\hline
 & \multicolumn{10}{c|}{{Poisson Arrivals per Period}}\\ 
 & \multicolumn{10}{c|}{}\\
      & \multicolumn{2}{c|}{100}     & \multicolumn{2}{c|}{500} &  \multicolumn{2}{c|}{5000} & \multicolumn{2}{c|}{10000} & \multicolumn{2}{c|}{20000}   \\ 
$T$      & $\Delta$AIC    & $\Delta$BIC     & $\Delta$AIC    &   $\Delta$BIC          &$\Delta$AIC      & $\Delta$BIC &$\Delta$AIC & $\Delta$BIC &$\Delta$AIC & $\Delta$BIC \\ \hline\hline
 50   & 62.88  & -406.49 & 647.80 &   63.45     & 6573.78 & 5822.90 & 12721.63 & 11920.82 & 27213.08  & 26362.35  \\
 100  & 52.17  & -467.17 & 590.34 &  -43.95    & 6968.28 & 6167.35 & 12883.76 & 12032.97 & 25498.83 & 24598.21  \\
 500  & 54.34  & -580.98 & 568.20 &  -182.82   & 6630.19 & 5713.51 & 12838.87 & 11872.31 & 26791.22 &  25774.68 \\
1000  & 49.71  & -635.30 & 614.09 &  -186.59    & 6528.87 & 11500.41 & 12516.90 & 25663.89 & 26730.29 &  22696.51 \\
 5000 & 45.86  & -754.92 & 568.29 &  -348.42    & 6771.89 & 11574.90 & 12707.28 & 23777.10 & 24959.38  & 23761.85 \\
10000 & 23.55  & -827.10 & 527.66 &  -438.95    & 7228.27 & 6095.91 & 13276.66 & 12094.37 & 24923.98 &  23691.78  \\ \hline
\end{tabular}
\label{table:AIC-BIC}
\end{table}

In a next set of simulations, we  generated data from the Halo-MNL model (the parameters used for the simulation can be found in the appendix). Then, we fit both MNL and Halo-MNL models in order to observe their behavior when interaction effects are indeed present. Likelihood, AIC, and BIC scores are summarized in Table \ref{table:Halo-sim}. Based on these, we can conclude that  the Halo-MNL has a lower Likelihood score and better  AIC and BIC scores when the sample size is sufficiently large that allows all parameters be estimated.

\begin{table}
\centering
\caption{Data generated by Halo-MNL:  $\Delta AIC = AIC_{MNL} - AIC_{Halo-MNL}$ \& $\Delta BIC = BIC_{MNL} - BIC_{Halo-MNL}$ $\Delta L = Likelihood_{MNL} - Likelihood_{Halo-MNL}$ }
\begin{tabular}{|c | c c c | c c c| c c c |} 
 \hline
  &  \multicolumn{9}{c}{Mean Poisson Arrivals Per Period} \\
 {} &  \multicolumn{3}{c}{100}  & \multicolumn{3}{c}{500}  & \multicolumn{3}{c|}{10,000}  \\
T  &  $\Delta L$ &  $\Delta AIC$ & $\Delta BIC$ &  $\Delta L$ & $\Delta AIC$ & $\Delta BIC$ &  $\Delta L$ & $\Delta AIC$ & $\Delta BIC$ \\
 \hline\hline
 50    &  -127.7 & 111.3 & -358.3 & -520.9 & 897.7 & 312.6 &  -11254.3 & 22364.6 & 21563.8\\ 
 100   &  -293.1 & 442.1 & -77.9 &  -1135.4 & 2126.8 & 1491.5 & -22704.4 & 45264.7 & 44413.9 \\
 500  & -1138.3 & 2132.6 & 1497.5 & -5593.5 & 11043.0 & 10292.1 & -113825.7 & 227507.3 & 226540.7\\
 1000& -2218.7 & 4293.4 &  3608.7& -11236.9 & 22329.9 & 21529.2 & -227258.4 & 454372.7 & 453356.2\\
 5000  & -11332.6 & 22521.2 & 21720.3 & -56438.5 & 112733 & 111816.3 & -1138958 & 2277771 & 2276639\\
 10000 & -23243.3 & 46342.6 & 45491.9 & -114161.5 & 228178.9 & 227212.3 & -2278759 & 4557375 & 4556193\\ 
 \hline
\end{tabular}
\label{table:Halo-sim}
\end{table}

\subsection{Application to Hotel Data}

In order to study the applicability of the Halo MNL model in a real dataset, we consider the Hotel Data  in \cite{Bodea2009Choice-BasedChain} that includes  room purchases with check-in dates between March 12, 2007 and April 15, 2007 for five different hotels.  For each hotel, the dataset  consists of  attributes, such as  booking and check-in date, customer reward status, room type and room amenities that determine the final price. For our purposes, we use hotel number 3 and we choose `room type' as our `product/item'. After cleaning the data, we have 9 distinct products all summarized in Table \ref{table:data}.

\begin{table}[!h]
\centering
	\caption{Products and their labels}
\begin{tabular}{lc}
\textbf{Room Type} &  \textbf{Label}   \\
King Room 1 Smoking  & Product 1  \\
King Room 3 Non-Smoking  & Product 2 \\
2 Double Beds Room 1 Smoking  & Product 3  \\
2 Double Beds Room 2 Smoking  & Product 4\\
Regular Bed Room 1 Smoking  & Product 5 \\
Regular Bed Room 1 Non-Smoking  & Product 6\\
Standard Room  &  Product 7    \\
King Room 1 Non-Smoking  &  Product 8        \\
2 Double Beds Room 1 Non-Smoking  &  Product 9 \\
\end{tabular}
\label{table:data}
\end{table}

\subsection*{Goodness of Fit for MNL and Halo-MNL}
The goal in this section is to fit both the MNL and Halo MNL models in the hotel data set. We normalize the utility of `Product 9' to be zero in order to be consistent with our model assumptions in Section 4. 
The final dataset consists of 834 periods (offer sets) that satisfy condition $(\mathcal{C}1)$.
We split the dataset in training and testing, considering different percent allocations in order to observe the effect of sample size  in  model performance. Specifically, we consider training sets of sizes 200, 300, 400, 500, and 600  that are randomly selected, with the remaining offer sets being part of the testing set. Due to the random split of the dataset in training and testing, the  sets that we work will only satisfy condition $(\mathcal{C}2)$, which  guarantees partial identifiability. In layman's terms, this translates into the fact that we are only able to estimate parameters for which we have available observations. This is an issue that is more prominent in the Halo MNL model, since we require more data to estimate all the  parameters.

For the Halo-MNL model we use the  maximum likelihood estimates from the previous section, while for the estimation of the MNL parameters we use the Baron solver for a brute force maximization of the likelihood.

\noindent \textit{Remark:} The   Branch-And-Reduce Optimization Navigator (BARON) 17.8.9 solver  is used (\cite{Sahinidis2015BARONPrograms}, \cite{Tawarmalani2005AOptimization}) under Pyomo package in Python (\cite{Hart2017PyomoPython}, \cite{Hart2011Pyomo:Python}). Specifically, BARON is a computational system for finding the global solution for non-linear non-convex programs. In all computations in Baron, the iterations are terminated when the approximate solution is not improved by at least $10^{-12}$ in absolute value.

 After fitting both models, we first compute the  reward gained by using each one for correctly classifying whether an item was purchased or not. Specifically, we compute the  logarithm of the predicted probability of an item chosen in the transaction, given the offer set. The sum of the indices for test sets of different sizes is summarized in Table \ref{table:index}.
\begin{table}[!h]
\centering
\caption{Reward Index for MNL and Halo-MNL models on different sizes of datasets.}
\begin{tabular}{lccccc}
\hline
\multirow{2}{*}{\textit{Model}} & \multicolumn{5}{c}{\textit{Training Set Sample Size}}\\
      & \textbf{200}  & \textbf{300}  & \textbf{400} & \textbf{500} & \textbf{600} \\ \hline \hline
\textbf{MNL}  & 422.599 & 384.542 & 342.371 & 299.531 & 254.825 \\
\textbf{Halo MNL}  & 404.374 & 360.835 & 342.419 & 310.092 & 259.404 \\ \hline
\end{tabular}
\label{table:index}
\end{table}

Based on this index, we observe that the proposed Halo-MNL model has a weaker performance than the MNL for smaller sizes of training sets, while for sufficient large training sets, both models perform equally well with the Halo-MNL model being slightly better.  

Furthermore, we perform a goodness of fit  to the multinomial distribution test for both models. For the purpose of the goodness of fit test, we only consider a training set consisting of 500 periods. We use maximum likelihood to fit both models into consideration and we compute the chi-square test statistic. In order to compute the corresponding $p$-values, we need to have a sufficiently large dataset relative to the estimated probabilities of purchase (for each offer set separately). Statistically, this condition translates to requiring $n\;p_{i}>10$ $\forall i= 1, 2, ..., N$, where $n$ is the number of times a specific offer set appears in the data and $p_i$ is the  probability of purchasing item $i$ given that this offer set. Unfortunately, this condition is not met in most cases of the offer sets, therefore we result in computing the $p$-values using bootstrap.

The bootstrap approach essentially relies on random sampling with replacement, in order to capture the variability of the  underlying population. Typically the $p$-value from a sample is a fixed number. However, with random re-sampling with replacement, we are able to obtain a new sample at every iteration, eventually leading to a different $p$-value. In our case, we randomly sample 1000 times with replacement  from the offer sets that correspond to the same availability vector, each time calculating the $p$-value. In  Table \ref{table:boot-p} we record the  median $p$-value   both for the MNL  and Halo-MNL models.
\begin{table}[!h]
\centering
\caption{Median Bootstrap p-values corresponding to different offer sets.}
\begin{tabular}{lcccc}
\hline
\multirow{2}{*}{\textit{Model}} & \multicolumn{4}{c}{\textit{Offer Set ID}}\\ 
 & \textbf{no. 1} & \textbf{no. 2} & \textbf{no. 3} & \textbf{no. 4}\\ \hline \hline
\textbf{MNL}      & 0.342 & 0.984 & 0.210 & 0.189 \\
\textbf{Halo-MNL} & 0.000   & 1.000 & 0.641 & 0.263 \\ \hline
\end{tabular}
\label{table:boot-p}
\end{table}
Based on these results, we observe that in three out of four different offer sets, the Halo-MNL outperforms the MNL in terms of how the model fits the data, resulting in a higher $p$-value. 
 
For this example, if we carefully study the estimated probabilities based on both models, we observe that the probability of purchase of an item is decreased, when another competing product is removed from the offer set. For example, if we compare an offer set when only `item 2' (for example) is missing, with an offer set where no items are missing, we observe that the probability of purchasing `item 8', for example, is lower in the first case. This is an interesting situation where the absence of an item from the offer set leads to a decrease in the probability of others being purchased. Therefore, a model capturing positive and negative effects among products is needed.

The results for log-likelihood, AIC, and BIC scores for the testing set are summarized in Table \ref{table:testing}.
\begin{table}[!h]
\centering
\caption{Model comparison in Testing data set; Hotel Data.}
\begin{tabular}{lccc}
\hline
 \textit{Model}  & \textbf{Likelihood Score}  & \textbf{AIC Score}  & \textbf{BIC Score} \\ \hline \hline
\textbf{MNL}      & -304.43 & 624.83 & 699.49 \\
\textbf{Halo-MNL} & -301.34 & 666.68 & 965.34 \\ \hline
\end{tabular}
\label{table:testing}
\end{table}

According Table \ref{table:testing}, the likelihood of the testing set under Halo MNL model is higher than MNL model, however the AIC and BIC scores of the MNL model is better. This was predictable according to the simulations done before, since it is shown that as the number of offer sets in the training set increases, the Halo-MNL model tends to be better than MNL model in AIC and BIC scores; and for small number of offer sets, MNL model performs better. 

Finally, we want to compare the likelihood, AIC and BIC scores for both models in the full dataset. For this reason, we use Baron in both cases to compute the maximum likelihood estimates. The CPU time of run for MNL model was 6.4 seconds and for our own model was 20.48 seconds. The results are summarized in Table \ref{table:full}. We observe that under all metrics, and despite the larger number of parameters, the Halo-MNL has a significantly better fit to the data than the MNL.

\begin{table}[!h]
\centering
\caption{Model comparison in Full dataset.}
\begin{tabular}{lccc}
\hline
 \textit{Model}  & \textbf{Likelihood Score}  & \textbf{AIC Score}  & \textbf{BIC Score} \\ \hline \hline
\textbf{MNL}      & -1259.91 & 2535.82 & 2576.17 \\
\textbf{Halo-MNL} &  -1192.10 & 2404.20 & 2454.64 \\ \hline
\end{tabular}
\label{table:full}
\end{table}

\section{Conclusion}
Our main goal in this paper was to suggest a novel customer choice model in the assortment selection framework that captures  positive and/or negative pairwise effects between products that are present or absent from a collection. Therefore, we propose a generalization of the classical MNL model that includes interaction effects. In this way, when assessing the customer preferences, we are able to understand the synergistic (or not) effect of items that are absent from an assortment.

Furthermore, since the proposed model has a significantly higher number of parameters to estimate, compared to the MNL, we derived identifiability and partial identifiability conditions that will ensure the validity of the parameter estimates. Following a maximum likelihood approach, we proposed closed-form solutions for the parameter estimators in specific scenarios.

Finally, we studied the numerical performance of our model in a simulated environment, as well as in a read data example.  Based on our numerical studies, we observed that when sufficiently enough data exist, the Halo-MNL model outperforms the classical MNL in terms of how well it fits the data into consideration. This is indeed true, even when the metrics used penalize  for the number of parameters in the model.

In conclusion, this is a first attempt in the literature to model \textit{positive} interaction effects in customer choice models in the context of assortment planning. There are, of course, several questions that still need to be addressed, such as revenue management  under such models,  models that have a richer dependence structure than the  proposed one, as well as efficient  statistical methods for parameter estimation, but these are left for future work.

\section{Appendix}

True Parameters for the Halo MNL model for the first set of simulations in Figure \ref{fig:error} and Table \ref{table:AIC-BIC}:

\[
A=
  \begin{bmatrix}
    0     & -0.241 & 0.182 & 0.105 & 0.182 & 0.095 & 0.140 & 0.134 & 0.182 & 0\\
   -0.236 & 0      & 0.110 & 0.125 & 0.134 & 0.049 & 0.095 & 0.069 & 0     & 0\\
    0.095 & 0.095  & 0     &-0.069 & 0.082 & 0.095 & 0.095 & 0     & 0.095 & 0\\
    0.086 & 0.042  &-0.143 & 0     &-0.059 & 0.049 & 0.049 & 0.069 & 0     & 0\\
    0.049 & 0.082  & 0.033 &-0.022 & 0     & 0.025 & 0.00  & 0     & 0.095 & 0\\
    0.077 & 0.095  & 0.017 & 0.065 &-0.121 & 0     &-0.105 & 0.069 & 0.182 & 0\\
    0.030 & 0.069  & 0.033 & 0.065 & 0.056 & 0.072 & 0     & 0     & 0     & 0\\
    0.030 & 0.056  & 0.065 & 0.085 & 0.082 & 0     &-0.051 & 0     & 0.262 & 0\\
    0.039 &-0.029  & 0.033 & 0.022 & 0     & 0.049 & 0     & 0.069 & 0     & 0\\
    0     & 0      & 0     & 0     & 0     & 0     & 0     & 0     & 0     & 0\\
  \end{bmatrix}
\]
\[\mu = [0, -0.357, -0.511, -0.799, -1.050, -0.916, -1.609, -1.966, -2.303, 0]\]

True Parameters for the Halo MNL model for the second set of simulations in Section 5.1:

\[
A=
  \begin{bmatrix}
    0     & 0.055& -0.094& -0.035& -0.256&  0.458& -0.095& -0.305& -0.183& 0\\
    0.137 &  0   & -0.330& -0.255& -0.115& -0.219&  0.210& -0.404&  0.380& 0\\
    0.473 &-0.070&  0    & -0.190&  0.239&  0.458& -0.295&  0.286& -0.479& 0\\
    0.417 & 0.365&  0.375&  0    &  0.020&  0.472&  0.499& -0.214&  0.266& 0\\
   -0.162& 0.390 & -0.394& -0.465&  0    & -0.014& -0.328& -0.119& -0.307& 0\\
   -0.112&-0.239 &  0.131&  0.231&  0.373&  0    &  0.088& -0.244&  0.374& 0\\
    0.479&-0.427 & -0.320& -0.343&  0.346&  0.301&  0    &  0.169&  0.281& 0\\
    0.276& 0.238 & -0.053& -0.075&  0.403& -0.291&  0.057&  0    &  0.466& 0\\
    0.215&-0.176 &  0.356&  0.317&  0.013& -0.086&  0.023&  0.287&  0    & 0\\
    0    & 0    &  0    &  0    &  0    &  0    &  0    &  0    &  0     & 0\\
  \end{bmatrix}
\]
{\small
$\mu = (-1.050, -0.565, -0.130, -0.225, -1.210, -0.904, -1.031, -0.755, -0.588,  0)$
}

\bibliographystyle{plain}
\bibliography{halo-refs.bib}

\end{document}